\documentclass[12pt]{article}
\usepackage{amsmath}
\usepackage{amssymb}
\usepackage{amsthm}
\usepackage{graphicx}
\usepackage{lineno}
\usepackage{float}
\usepackage[all]{xy}
\usepackage{mathrsfs}
\usepackage{gensymb}

\usepackage{xcolor}
\usepackage{tikz}
\usetikzlibrary{tqft}

\newtheorem{theorem}{Theorem}

\newtheorem{corollary}{Corollary}

\newcommand{\be}{\begin{equation}}
\newcommand{\ee}{\end{equation}}

\newcommand{\bse}{\begin{subequations}}
\newcommand{\ese}{\end{subequations}}
\newcommand{\bea}{\begin{eqnarray}}
\newcommand{\eea}{\end{eqnarray}}
\newcommand{\ba}{\begin{array}}
\newcommand{\ea}{\end{array}}
\newcommand{\bc}{\begin{center}}
\newcommand{\ec}{\end{center}}

\usepackage{braket}

\usepackage{amsmath}
\usepackage{amssymb}
\usepackage{amsthm}
\usepackage{graphicx}
\usepackage{lineno}
\usepackage{float}
\usepackage{mathrsfs}

\usepackage{xcolor}
\usepackage{tikz}
\usetikzlibrary{tqft}
\usepackage[compat=1.1.0]{tikz-feynman}
\usetikzlibrary{decorations.pathmorphing}
\usetikzlibrary{decorations.markings}

\begin{document}

\title{\bf  ER = EPR is an operational theorem}

\author{{Chris Fields$^1$, James F. Glazebrook$^2$, Antonino Marcian\`{o}$^3$ and Emanuele Zappala$^4$}\\ \\
{\it$^1$ Allen Discovery Center}\\
{\it Tufts University, Medford, MA 02155 USA}\\
{fieldsres@gmail.com} \\
{ORCID: 0000-0002-4812-0744} \\ \\
{\it$^2$ Department of Mathematics and Computer Science,} \\
{\it Eastern Illinois University, Charleston, IL 61920 USA} \\
{\it and} \\
{\it Adjunct Faculty, Department of Mathematics,}\\
{\it University of Illinois at Urbana-Champaign, Urbana, IL 61801 USA}\\
{jfglazebrook@eiu.edu}\\
{ORCID: 0000-0001-8335-221X}\\ \\
{\it$^3$ 
Center for Astronomy and Astrophysics} \\
{\it \& 
Center for Field Theory and Particle Physics } \\
{\it \& Department of Physics, Fudan University, Shanghai, CHINA} \\
{\it and} \\
{\it Laboratori Nazionali di Frascati INFN, Frascati (Rome), Italy, EU
} \\
{\it and} \\
{\it INFN sezione di Roma ``Tor Vergata'', 00133 Rome, Italy, EU}\\
{marciano@fudan.edu.cn} \\
{ORCID: 0000-0003-4719-110X} \\ \\
{\it$^4$ Department of Mathematics and Statistics,} \\
{\it Idaho State University, Pocatello, ID 83209, USA} \\
{emanuelezappala@isu.edu} \\
{ORCID: 0000-0002-9684-9441} \\ \\
}



\date{\today}

\maketitle

{\bf Abstract:} 
We show that in the operational setting of a two-agent, local operations, classical communication (LOCC) protocol, Alice and Bob cannot operationally distinguish monogamous entanglement from a topological identification of points in their respective local spacetimes, i.e. that ER = EPR can be recovered as an operational theorem. Our construction immediately implies that in this operational setting, the local topology of spacetime is observer-relative.  It also provides a simple demonstration of the non-traversability of ER bridges.  As our construction does not depend on an embedding geometry, it generalizes previous geometric approaches to ER = EPR.

\section{Introduction}
\noindent 
Maldecena and Susskind \cite{maldecena:13} conjectured, using geometric arguments in AdS/CFT applied to black holes, that a wormhole, i.e. an Einstein-Rosen (ER) bridge, is equivalent to a pair of maximally Einstein-Podolski-Rosen (EPR) entangled black holes, i.e. that ER=EPR (see also \cite{susskind:16}). 
Here we show that two distinct observers cannot operationally distinguish, by independent local manipulations and measurements, monogamous entanglement from a topological identification of points in their respective local spacetimes. In other words, in a two-agent setting it is not possible to operationally distinguish ER from EPR.  Our result represents a rigorous derivation of ER=EPR as an operational theorem; furthermore, it provides a  generalization of ER=EPR, with no necessity of geometric embedding, and a demonstration of the non-traversability of ER bridges.

The relevance of ER=EPR lies in its possible resolution to the AMPS firewall paradox, as introduced in \cite{Almheiri:2012rt}. The eventual appearance of a firewall can be traced to the flow of energy/matter among the two distant black holes. But the appearance of the firewall inside the event horizon prohibits sending any external superluminal signal.  Maldacena and Susskind showed that pair production of charged black holes in a background magnetic field triggers their entanglement, and that after Wick rotation these behave as a wormhole.  They extended then their conjecture so as to comprise not only entangled pairs of black holes, but also any possible entangled pairs of particles. 
These particles would be also connected by wormholes, at higher (Planckian) energy scales. This opened the pathway to a fascinating interpretation of the conjecture, that spacetime geometry and entanglement are interwoven, hence substantiating a previous observation by Van Raamsdonk 
\cite{VanRaamsdonk:2010pw}. 

From an operational perspective, measurements of entangled pairs implement local operations, classical communication (LOCC) protocols \cite{chitambar2014everything}. This is clear in a canonical Bell/EPR experiment, where the agents Alice and Bob must agree, via classical communication, to employ specified detectors in specified ways, and must later exchange their accumulated data (or transfer it to some 3rd party) in the form of classical records.  It has been shown in \cite{fgm:22a} that sequentially-repeated state preparations and/or measurements that employ mutually commuting quantum reference frames (QRFs \cite{bartlett:07}), e.g. the sequentially repeated preparations and/or measurements of position and spin during a Bell/EPR experiment, can be represented, without loss of generality, by topological quantum field theories (TQFTs \cite{atiyah:88}). It was then shown in \cite{fgm:24a} that any two-agent LOCC protocol can be represented by Diagram \eqref{locc-diag}, in which the agents Alice and Bob are i) mutually separable, and hence conditionally statistically independent, ii) are separated from their joint environment $E$ by a holographic screen $\mathscr{B}$, iii) implement read/write QRFs $Q_A$ and $Q_B$, respectively, and iv) communicate via classical and quantum channels implemented by $E$.

Intuitively, Alice and Bob are connected by a quantum channel if there exist distinct (collections of) qubits $q_A$ and $q_B$ accessible only to Alice and Bob, respectively, and $|q_Aq_B\rangle \neq |q_A\rangle |q_B\rangle$.  Alice and Bob are connected by a classical channel if there exist distinct (collections of) qubits $q^{\prime}_A$ and $q^{\prime}_B$ accessible only to Alice and Bob, respectively, and causal processes $f$ and $g$ implemented by $E$ such that $q^{\prime}_B = f(q^{\prime}_A)$ and $q^{\prime}_A = g(q^{\prime}_B)$.  Here ``causal'' means causal with respect to clocks in $E$. Alice and Bob being distinct, mutually separable agents that communicate classically --- i.e. causally --- entails that the QRFs $Q_A$ and $Q_B$ do not commute, and hence that LOCC protocols exhibit quantum contextuality
%
\cite{fg:21,fg:23}. We can, therefore, consider Alice and Bob to also be separated by a boundary, which is elided in Diagram \eqref{locc-diag} to emphasize their joint interaction with $E$, and to have access only to distinct, non-overlapping sectors of $\mathscr{B}$.  In practical Bell/EPR experiments, the separation of Alice and Bob is effected by their spacelike separation in a laboratory coordinate system; in the notation of Diagram \eqref{locc-diag}, this corresponds to the classical channel being timelike, i.e. causal, in the laboratory (i.e. $E$) coordinate system.

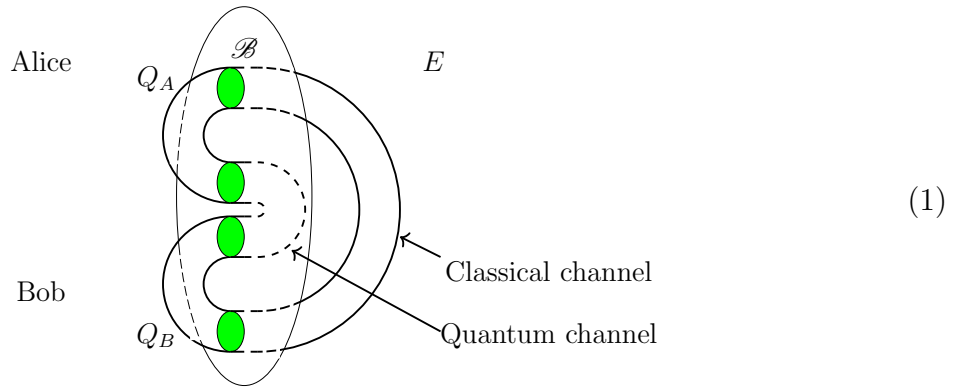
\begin{equation} \label{locc-diag}
\begin{gathered}
\scalebox{0.9}{
\begin{tikzpicture}[every tqft/.append style={transform shape}]
\draw[rotate=90] (0,0) ellipse (2.8cm and 1 cm);
\node[above] at (0,1.9) {$\mathscr{B}$};
\draw [thick] (-0.2,1.9) arc [radius=1, start angle=90, end angle= 270];
\draw [thick] (-0.2,1.3) arc [radius=0.4, start angle=90, end angle= 270];
\draw[rotate=90,fill=green,fill opacity=1] (1.6,0.2) ellipse (0.3 cm and 0.2 cm);
\draw[rotate=90,fill=green,fill opacity=1] (0.2,0.2) ellipse (0.3 cm and 0.2 cm);
\draw [thick] (-0.2,-0.3) arc [radius=1, start angle=90, end angle= 270];
\draw [thick] (-0.2,-0.9) arc [radius=0.4, start angle=90, end angle= 270];
\draw[rotate=90,fill=green,fill opacity=1] (-0.6,0.2) ellipse (0.3 cm and 0.2 cm);
\draw[rotate=90,fill=green,fill opacity=1] (-2.0,0.2) ellipse (0.3 cm and 0.2 cm);
\draw [rotate=180, thick, dashed] (-0.2,0.9) arc [radius=0.7, start angle=90, end angle= 270];
\draw [rotate=180, thick, dashed] (-0.2,0.3) arc [radius=0.1, start angle=90, end angle= 270];
\draw [thick] (-0.2,0.5) -- (0,0.5);
\draw [thick] (-0.2,-0.1) -- (0,-0.1);
\draw [thick] (-0.2,-0.9) -- (0,-0.9);
\draw [thick] (-0.2,-0.3) -- (0,-0.3);
\draw [thick, dashed] (0,0.5) -- (0.2,0.5);
\draw [thick, dashed] (0,-0.1) -- (0.2,-0.1);
\draw [thick, dashed] (0,-0.9) -- (0.2,-0.9);
\draw [thick, dashed] (0,-0.3) -- (0.2,-0.3);
\node[above] at (-3,1.7) {Alice};
\node[above] at (-3,-1.7) {Bob};
\node[above] at (2.8,1.7) {$E$};
\draw [ultra thick, white] (-0.9,1.5) -- (-0.7,1.5);
\draw [ultra thick, white] (-1,1.3) -- (-0.8,1.3);
\draw [ultra thick, white] (-1,1.1) -- (-0.8,1.1);
\draw [ultra thick, white] (-1,0.9) -- (-0.8,0.9);
\draw [ultra thick, white] (-1.1,0.7) -- (-0.8,0.7);
\draw [ultra thick, white] (-1.1,0.5) -- (-0.8,0.5);
\draw [ultra thick, white] (-1,-0.9) -- (-0.8,-0.9);
\draw [ultra thick, white] (-1,-1.1) -- (-0.8,-1.1);
\draw [ultra thick, white] (-1,-1.3) -- (-0.8,-1.3);
\draw [ultra thick, white] (-0.9,-1.5) -- (-0.7,-1.5);
\draw [ultra thick, white] (-0.9,-1.7) -- (-0.7,-1.7);
\draw [ultra thick, white] (-0.8,-1.9) -- (-0.6,-1.9);
\draw [ultra thick, white] (-0.8,-2.1) -- (-0.6,-2.1);
\node[above] at (-1.3,1.4) {$Q_A$};
\node[above] at (-1.3,-2.4) {$Q_B$};
\draw [rotate=180, thick] (-0.2,2.3) arc [radius=2.1, start angle=90, end angle= 270];
\draw [rotate=180, thick] (-0.2,1.7) arc [radius=1.5, start angle=90, end angle= 270];
\draw [thick] (-0.2,1.9) -- (0,1.9);
\draw [thick] (-0.2,1.3) -- (0,1.3);
\draw [thick, dashed] (0.2,1.9) -- (0,1.9);
\draw [thick, dashed] (0.2,1.3) -- (0,1.3);
\draw [thick] (-0.2,-1.7) -- (0,-1.7);
\draw [thick] (-0.2,-2.3) -- (0,-2.3);
\draw [thick, dashed] (0.2,-1.7) -- (0,-1.7);
\draw [thick, dashed] (0.2,-2.3) -- (0,-2.3);
\draw [ultra thick, white] (0.3,2) -- (0.3,1.2);
\draw [ultra thick, white] (0.5,2) -- (0.5,1.2);
\draw [ultra thick, white] (0.7,1.9) -- (0.7,1.1);
\draw [ultra thick, white] (0.3,-2.4) -- (0.3,-1.5);
\draw [ultra thick, white] (0.5,-2.4) -- (0.5,-1.5);
\draw [ultra thick, white] (0.7,-1.8) -- (0.7,-1.5);
\node[above] at (4.5,-2.4) {Quantum channel};
\draw [thick, ->] (2.9,-2) -- (0.7,-0.8);
\node[above] at (4.5,-1.4) {Classical channel};
\draw [thick, ->] (2.9,-0.9) -- (2.3,-0.6);
\end{tikzpicture}
}
\end{gathered}
\end{equation}

The tokens via which Alice and Bob classically communicate --- e.g. modulations of the ambient photon field --- comprise degrees of freedom of $E$ and are therefore quantum systems that must be measured to extract state information \cite{tipler:14}. The distinction between classical and quantum channels is, therefore, an a priori assumption that allows defining LOCC, not an observational outcome of either Alice's or Bob's interactions with $E$. Nonetheless, a LOCC protocol permits Alice and Bob to determine via observation and classical communication if a quantum channel is shared between them. This is empirically manifest in Bell/EPR experiments, and relied upon theoretically when treating quantum Darwinism \cite{zurek:09} as enabling a ``public'' quantum-to-classical transition \cite{fgm:24a}. 

The detection of monogamous entanglement, i.e. entanglement reaching the relevant Tsirelson bound \cite{tsirelson:80} for violation of the relevant Bell inequality, requires a perfectly decoherence-free quantum channel traversing $E$. Thus, we can also state our result as showing that Alice and Bob cannot operationally distinguish a perfectly decoherence-free quantum channel traversing $E$ from a topological identification of points in their respective measured spacetimes.  This result is demonstrated below using a purely topological argument, and is independent of the geometry and coordinates (i.e. spatial QRFs) employed by either Alice or Bob to describe either the boundary $\mathscr{B}$ or the bulk $E$.

A number of criticisms of ER = EPR on geometric grounds \cite{dai:20} become, in this setting, constraints on Alice's and Bob's abilities to measure and communicate about spacetime, i.e. constraints on the QRFs they employ to measure space and time.  As entanglement is known to be observer/description/QRF-relative \cite{zanardi:00, zanardi:04}, our result immediately implies that the local topology of spacetime is observer/description/QRF-relative.

\section{LOCC protocols}

\subsection{Quantum instruments and quantum channels} 
\noindent 
We first briefly recall the definition of LOCC, closely following \cite{chitambar2014everything}. Recall that given a quantum system with corresponding Hilbert space $\mathcal{H}$, a {\it quantum instrument} is a family of completely positive maps $\mathcal E_j : B(\mathcal{H}) \longrightarrow B(\mathcal{H})$, where each $\mathcal E_j$ is a bounded linear map on the space of bounded linear operators $B(\mathcal{H})$ of $\mathcal{H}$. Additionally, it is required that the sum $\sum_j \mathcal E_j$ is trace preserving, and there are countably (possibly finitely many) indices $j$.
The set of quantum instruments is in one-to-one correspondence with the set of quantum-classical maps, defined as trace-preserving completely positive (TCP) maps $B(\mathcal{H}) \longrightarrow B(\mathcal{H})\otimes B(\mathbb C^\Theta)$ of the form $\rho \mapsto \sum_j \mathcal E_j(\rho) \otimes \ket{j}\bra{j}$, where $\Theta$ denotes the index set of quantum instruments. This leads to characterizing the actual nature of a quantum channel. 

In two-agent protocols as depicted in Diagram \eqref{locc-diag} above, $j = 2$ and each of the QRFs $Q_A$ and $Q_B$ corresponds to the composition of a quantum instrument acting on the relevant boundary sector of the quantum channel with a quantum-classical map acting on the relevant boundary sector of the classical channel. 
An $n$-partite system has a corresponding Hilbert space given by the tensor product $\mathcal{H} = \mathcal{H}^{k_1}\otimes \cdots \otimes \mathcal{H}^{k_n}$, where $\mathcal{H}^{\otimes k_j}$ represents the reduced state space of the $j^{\rm th}$ component of the system. In this case, an instrument $\{\mathcal F_j\}_{j\in \Theta}$ is one way local with respect to the $k^{\rm th}$ component if we have $\mathcal F_j = \bigotimes_{j\neq k} \mathcal T_j\otimes \mathcal E_k$, where $\mathcal E_k$ is a completely positive map while each $\mathcal T_j$ is TCP. This operation represents party $k$ applying a quantum instrument $\mathcal E_i$ and then transmitting the outcome to all other parties classically. 

\subsection{Implementing LOCCs} 
\noindent
An instrument $\mathfrak I'$ is LOCC linked to $\mathfrak I$ if $\mathfrak I'$ is a coarse-graining of $\mathfrak I$, where coarse-graining is, roughly speaking, a procedure of grouping instruments in $\mathfrak{I}$ in a compatible way (see \cite{chitambar2014everything}). Then, LOCCs are defined recursively. One says that $\mathfrak I$ is ${\rm LOCC}_1$ if it is local with respect to some party. One says that $\mathfrak I$ is ${\rm LOCC}_n$ if it is LOCC linked to some $\mathfrak J$ which is ${\rm LOCC}_{n-1}$. {\rm LOCC} is defined as the direct limit of a system of ${\rm LOCC}_n$ instruments.  For the present purposes, only two-observer, i.e. ${\rm LOCC}_2$ protocols need be considered.  In the language of \cite{fgm:24a} and hence of Diagram \eqref{locc-diag}, this LOCC linking corresponds to the requirement that only mutually-commuting QRFs, e.g. position and spin in a Bell/EPR experiment, can be ``co-deployed'' by being combined into a single effective QRF. Formally, combining QRFs is simply concatenating commuting operators.  Any QRF can be represented as an operator by a Cone-CoCone Diagram (CCCD) of commuting component operators that collectively implement a logically-gated, distributed information flow, and out of which the TQFT is constructed (see \cite{fgm:22a,fgm:24a}, and references therein for details).  Concatenating commuting QRFs is then concatenating CCCDs, which commute by definition if the operators they represent commute.

\section{Bell/EPR experiments}
\noindent
Let us proceed to explain matters in terms of a canonical Bell/EPR experiment. In such an experiment, Alice and Bob can only detect entanglement if they each have a free choice of measurement basis, i.e. only if they can each independently choose the instrument settings to employ for each measurement.  If this free choice requirement is violated by some form of superdeterminism (or ``conspiracy''), Alice and Bob are themselves effectively entangled and cannot function as two independent observers.  It is also, clearly, nonsensical to talk about ``classical communication'' if Alice and Bob cannot be considered separate systems.  Hence formally, separability of the joint state, $|AB \rangle = |A \rangle |B \rangle$ (or of the joint density $\rho_{AB} = \rho_A \rho_B$) is a requirement of LOCC, and hence of operational access to a shared quantum channel \cite{fgm:24a}.

In Diagram \eqref{locc-diag}, both classical and quantum channels are implemented by $E$ and hence are, in general, exposed to interaction with other qubits contained within $E$. 
Interactions between the degrees of freedom of $E$ that implement the channels and other, non-channel degrees of freedom of $E$ implement eavesdropping or noise injection in the classical channel and decoherence in the quantum channel. In both cases, the coupling of channel to non-channel degrees of freedom is a quantum interaction, and its result is to transform pure states of channel degrees of freedom into mixed states.  If such mixing occurs, observations of the channel states, whether classical or quantum, will be characterized by stochastic noise, with uniform noise spectra in the limit of maximally mixed states. Alice and Bob have no direct observational access to such mixing interactions, which occur entirely within $E$ and are implemented by the internal interaction Hamiltonian $H_E$. The risk of decoherence can be minimized by minimizing the exposure of the quantum channel to the rest of $E$. 
Formally, such procedures correspond not to tracing out degrees of freedom of $E$ but to setting their interaction with the channel degrees of freedom to zero.  

\section{ER = EPR}
\subsection{Alice and Bob interacting via a quantum channel}
\noindent 
Consider now a quantum channel in the limit of zero decoherence.  In an idealized Bell/EPR experiment, this corresponds to monogamous pairwise entanglement, and hence to a Bell inequality violation reaching the Tsirelson bound, provided Alice and Bob consistently choose settings $45\degree$ apart.  Setting decoherence equal to zero is requiring that there no interaction between channel and non-channel degrees of freedom of $E$.  Hence the limit of zero decoherence in the quantum channel is reached as the interaction, in $E$, between channel and non-channel degrees of freedom approaches zero.  Taking $q_A$ and $q_B$ to be the (collections of) qubits accessible exclusively to Alice and Bob, respectively as above, the limit of zero decoherence is reached when the state of the quantum channel is simply the pure state $|q_Aq_B \rangle$.

We can assume, without loss of generality, that Alice and Bob interact directly with the qubits $q_A$ and $q_B$ to which they have exclusive access.  We can, therefore, assume that these qubits are localized to Alice's and Bob's respective sectors of the boundary $\mathscr{B}$.  The pure state $|q_Aq_B \rangle$ is, in this case, an entangled state of $\mathscr{B}$.  

This can be accomplished as shown in Diagram \eqref{top-trans}, where for convenience the boundary $\mathscr{B}$ is depicted edge-on. The operations on $\mathscr{B}$ depicted from left to right in Diagram \eqref{top-trans} do not change the topological relationships between Alice and Bob or between either Alice or Bob and $E$, and do not change $E$; they merely take the number of degrees of freedom of $E$, required to implement the quantum channel to zero. The joint state $|AB \rangle$ remains separable, as is required for Alice and Bob to have free choice of basis (i.e. to be regarded as experimenters) and for their interaction to remain LOCC. We assume that the depicted operations on $\mathscr{B}$ constitute a homotopy of the quantum channel which leaves the rest of the space fixed, i.e. it is the identity map outside of the quantum channel. So, in particular, we assume to be in the situation where the instruments realizing $Q_A$ and $Q_B$ are unchanged during this operation.  

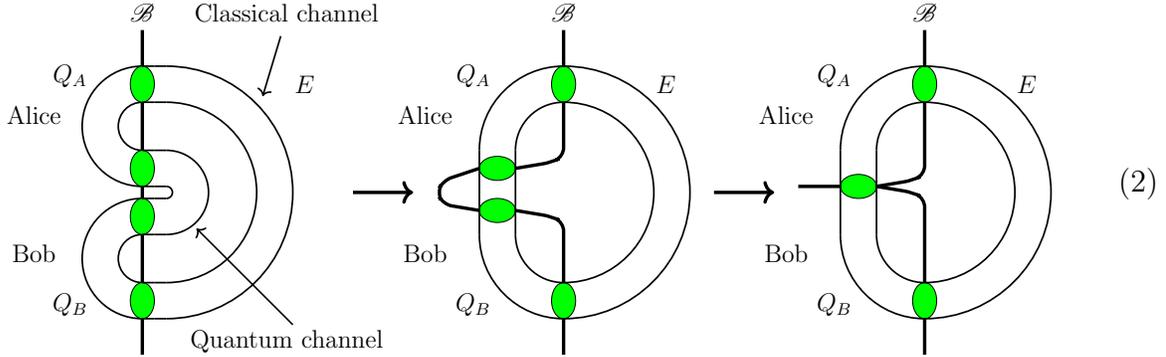
\begin{equation} \label{top-trans}
\begin{gathered}
\scalebox{0.8}{
\begin{tikzpicture}[every tqft/.append style={transform shape}]
\draw [ultra thick] (-0.2,2.5) -- (-0.2,1.9);
\draw [ultra thick] (-0.2,1.3) -- (-0.2,0.5);
\draw [ultra thick] (-0.2,-0.1) -- (-0.2,-0.3);
\draw [ultra thick] (-0.2,-0.9) -- (-0.2,-1.9);
\draw [ultra thick] (-0.2,-2.3) -- (-0.2,-2.9);
\node[above] at (-0.2,2.5) {$\mathscr{B}$};
\draw [thick] (-0.2,1.9) arc [radius=1, start angle=90, end angle= 270];
\draw [thick] (-0.2,1.3) arc [radius=0.4, start angle=90, end angle= 270];
\draw[rotate=90,fill=green,fill opacity=1] (1.6,0.2) ellipse (0.3 cm and 0.2 cm);
\draw[rotate=90,fill=green,fill opacity=1] (0.2,0.2) ellipse (0.3 cm and 0.2 cm);
\draw [thick] (-0.2,-0.3) arc [radius=1, start angle=90, end angle= 270];
\draw [thick] (-0.2,-0.9) arc [radius=0.4, start angle=90, end angle= 270];
\draw[rotate=90,fill=green,fill opacity=1] (-0.6,0.2) ellipse (0.3 cm and 0.2 cm);
\draw[rotate=90,fill=green,fill opacity=1] (-2.0,0.2) ellipse (0.3 cm and 0.2 cm);
\draw [rotate=180, thick] (-0.2,0.9) arc [radius=0.7, start angle=90, end angle= 270];
\draw [rotate=180, thick] (-0.2,0.3) arc [radius=0.1, start angle=90, end angle= 270];
\draw [thick] (-0.2,0.5) -- (0,0.5);
\draw [thick] (-0.2,-0.1) -- (0,-0.1);
\draw [thick] (-0.2,-0.9) -- (0,-0.9);
\draw [thick] (-0.2,-0.3) -- (0,-0.3);
\draw [thick] (0,0.5) -- (0.2,0.5);
\draw [thick] (0,-0.1) -- (0.2,-0.1);
\draw [thick] (0,-0.9) -- (0.2,-0.9);
\draw [thick] (0,-0.3) -- (0.2,-0.3);
\node[above] at (-2,0.8) {Alice};
\node[above] at (-2,-1.5) {Bob};
\node[above] at (2.5,1.3) {$E$};
\node[above] at (-1.4,1.4) {$Q_A$};
\node[above] at (-1.4,-2.4) {$Q_B$};
\draw [rotate=180, thick] (-0.2,2.3) arc [radius=2.1, start angle=90, end angle= 270];
\draw [rotate=180, thick] (-0.2,1.7) arc [radius=1.5, start angle=90, end angle= 270];
\draw [thick] (-0.2,1.9) -- (0,1.9);
\draw [thick] (-0.2,1.3) -- (0,1.3);
\draw [thick] (0.2,1.9) -- (0,1.9);
\draw [thick] (0.2,1.3) -- (0,1.3);
\draw [thick] (-0.2,-1.7) -- (0,-1.7);
\draw [thick] (-0.2,-2.3) -- (0,-2.3);
\draw [thick] (0.2,-1.7) -- (0,-1.7);
\draw [thick] (0.2,-2.3) -- (0,-2.3);
\node[above] at (2.2,-3.0) {Quantum channel};
\draw [thick, ->] (2.3,-2.4) -- (0.7,-0.8);
\node[above] at (2.2,2.5) {Classical channel};
\draw [thick, ->] (2.1,2.4) -- (1.8,1.4);
\draw [ultra thick, ->] (3.3,-0.2) -- (4.3,-0.2);
\draw[rotate=90,fill=green,fill opacity=1] (1.6,-6.8) ellipse (0.3 cm and 0.2 cm);
\draw[fill=green,fill opacity=1] (5.7,0.2) ellipse (0.3 cm and 0.2 cm);
\draw[fill=green,fill opacity=1] (5.7,-0.5) ellipse (0.3 cm and 0.2 cm);
\draw[rotate=90,fill=green,fill opacity=1] (-2.0,-6.8) ellipse (0.3 cm and 0.2 cm);
\draw [rotate=180, thick] (-6.8,2.3) arc [radius=2.1, start angle=90, end angle= 270];
\draw [rotate=180, thick] (-6.8,1.7) arc [radius=1.5, start angle=90, end angle= 270];
\draw [thick] (6.8,1.9) arc [radius=1.4, start angle=90, end angle= 180];
\draw [thick] (6.8,1.3) arc [radius=0.8, start angle=90, end angle= 180];
\draw [rotate=90,thick] (-0.9,-5.4) arc [radius=1.4, start angle=90, end angle= 180];
\draw [rotate=90,thick] (-0.9,-6.0) arc [radius=0.8, start angle=90, end angle= 180];
\draw [thick] (6.0,0.5) -- (6.0,-0.9);
\draw [thick] (5.4,0.5) -- (5.4,-0.9);
\draw [ultra thick] (6.8,2.5) -- (6.8,1.9);
\draw [rounded corners, ultra thick] (6.8,1.3) -- (6.8,0.4) -- (6.6,0.3) -- (6.0,0.2);
\draw [rounded corners, ultra thick] (5.4,0.2) -- (4.8,0.0) -- (4.7,-0.2) -- (4.8,-0.4) -- (5.4,-0.5);
\draw [rounded corners, ultra thick] (6.0,-0.5) -- (6.6,-0.6) -- (6.8,-0.7) -- (6.8,-1.7);
\draw [ultra thick] (6.8,-2.3) -- (6.8,-2.9);
\node[above] at (6.8,2.5) {$\mathscr{B}$};
\node[above] at (4.5,0.8) {Alice};
\node[above] at (4.5,-1.5) {Bob};
\node[above] at (8.5,1.3) {$E$};
\node[above] at (5.3,1.4) {$Q_A$};
\node[above] at (5.3,-2.4) {$Q_B$};
\draw [ultra thick, ->] (9.3,-0.2) -- (10.3,-0.2);
\draw[rotate=90,fill=green,fill opacity=1] (1.6,-12.8) ellipse (0.3 cm and 0.2 cm);
\draw[fill=green,fill opacity=1] (11.7,-0.1) ellipse (0.3 cm and 0.2 cm);
\draw[rotate=90,fill=green,fill opacity=1] (-2.0,-12.8) ellipse (0.3 cm and 0.2 cm);
\draw [rotate=180, thick] (-12.8,2.3) arc [radius=2.1, start angle=90, end angle= 270];
\draw [rotate=180, thick] (-12.8,1.7) arc [radius=1.5, start angle=90, end angle= 270];
\draw [thick] (12.8,1.9) arc [radius=1.4, start angle=90, end angle= 180];
\draw [thick] (12.8,1.3) arc [radius=0.8, start angle=90, end angle= 180];
\draw [rotate=90,thick] (-0.9,-11.4) arc [radius=1.4, start angle=90, end angle= 180];
\draw [rotate=90,thick] (-0.9,-12.0) arc [radius=0.8, start angle=90, end angle= 180];
\draw [thick] (12.0,0.5) -- (12.0,-0.9);
\draw [thick] (11.4,0.5) -- (11.4,-0.9);
\draw [ultra thick] (12.8,2.5) -- (12.8,1.9);
\draw [rounded corners, ultra thick] (12.8,1.3) -- (12.8,0.1) -- (12.6,0.0) -- (12.0,-0.1);
\draw [ultra thick] (11.4,-0.1) -- (10.7,-0.1);
\draw [rounded corners, ultra thick] (12.0,-0.1) -- (12.6,-0.2) -- (12.8,-0.3) -- (12.8,-1.7);
\draw [ultra thick] (12.8,-2.3) -- (12.8,-2.9);
\node[above] at (12.8,2.5) {$\mathscr{B}$};
\node[above] at (10.5,0.8) {Alice};
\node[above] at (10.5,-1.5) {Bob};
\node[above] at (14.5,1.3) {$E$};
\node[above] at (11.3,1.4) {$Q_A$};
\node[above] at (11.3,-2.4) {$Q_B$};
\end{tikzpicture}
}
\end{gathered}
\end{equation}

In a Bell/EPR experiment, the topological transformation effected in Diagram \eqref{top-trans} would be approximated by decreasing the laboratory-frame distance between each of Alice and Bob and the centrally-located source of entangled pairs toward zero, with the limit shown in the rightmost diagram representing a point source to which Alice and Bob are both immediately adjacent.  As discussed in \cite{fgm:21}, it also corresponds to a Bell/EPR experiment as described from the perspective of the entangled state $|q_Aq_B\rangle$, in which the observers Alice and Bob effectively collide at the fixed position of $|q_Aq_B\rangle$ (see \cite{fgm:21} Fig. 8a). 

\subsection{The main results}
\noindent
In leading to our main results reflected by the illustrative representation of Diagram \eqref{top-trans}, we will further clarify the setting. Here, we  explicitly assume that Alice, Bob, and $E$ are finite systems, and assume the holographic principle (HP) in geometrically-independent, background-free setting for all spaces in question. In particular, without loss of generality, the boundary $\mathscr{B}$ can be represented as an ancillary finite array of mutually non-interacting qubits, i.e. $\mathscr{B}$ has an effective Hilbert space $\mathcal{H}_{\mathscr{B}} = \otimes^N_{i=1} \mathcal{H}_{q_i}$ for some finite number $N$ of qubits $q_i$.  With this description, $\mathscr{B}$ is considered, in a strict sense, to be a distinct quantum system: it has no effect on the physics as implemented by the joint system, the self-interaction Hamiltonian $H_U$, or the $A-B$ interaction $H_{AB}$, or $H_{E}$.
The joint system $U=ABE$ thus requires $\mathcal{H}_U = 
\mathcal{H}_A \otimes \mathcal{H}_B \otimes \mathcal{H}_E$.  In the absence of any common embedding, this motivates having $\mathcal{H}_{\mathscr{B}} \cap \mathcal{H}_A \otimes \mathcal{H}_B \otimes \mathcal{H}_E = \emptyset$ \cite{fgm:22a, addazi:21, fgm:21, fgm:22b}.  The qubits $q_A$ and $q_B$ are, in this picture, components of $\mathscr{B}$ as discussed above. We can, with these assumptions, express the result of Diagram \eqref{top-trans} as a theorem:

\begin{theorem} \label{main-th}
In any LOCC protocol in which all systems are finite, and in which the boundary $\mathscr{B}$ between the communicating agents $A$ and $B$ and their joint environment $E$ is a holographic screen, as the entanglement made available to $A$ and $B$ by the quantum channel approaches pairwise monogamy, and hence the decoherence in the quantum channel detectable by $A$ or $B$ decreases to zero, the number of environmental degrees of freedom of $E$ required to implement the quantum channel becomes operationally indistinguishable, by $A$ or $B$, from zero in the limit of monogamous entanglement.
\end{theorem}

\begin{proof}
Suppose for convenience that $A$ and $B$ each access $N/2$ qubits of $\mathscr{B}$ and that both quantum and classical channels are informationally symmetric, i.e. the relevant sectors of $\mathscr{B}$ have the same dimensions for both $A$ and $B$.  Consider the case of what $A$ can detect; the case for $B$ is no different.  Let $A$ implement the QRFs and associated computations needed to prepare qubits within the sector of $\mathscr{B}$ corresponding to $A$'s end of the quantum channel (i.e. the qubit(s) $q_A$), according to information received via the classical channel, and send messages via the classical channel given observations of the qubits $q_A$; without loss of generality, we can consider these combined computations to implement a single QRF $Q_A$ \cite{fgm:24a}. $A$ can estimate decoherence in the quantum channel by comparing her measurement statistics, together with $B$'s statistics received via the classical channel, to the relevant Tsirelson bound; we are interested in the case in which the joint measurement statistics approaches the relevant Tsirelson bound and hence the estimated decoherence approaches zero.  

 Aside from this classical measure of purity, and hence decoherence, by the Tsirelson bound, $A$ can obtain information about the number of degrees of freedom of $E$ that implement the quantum channel only via the action of the Hamiltonian $H_E$ on qubits of $\mathscr{B}$ other than $q_A$.
Let $Q$ denote the degrees of freedom of $E$ that implement the quantum channel and $\bar{Q}$ denote its complement within $E$, i.e. $Q\bar{Q} = E$.  We can then write the Hamiltonian decomposition $H_E = H_Q + H_{\bar{Q}} + H_{Q\bar{Q}}$. The interaction $H_{Q\bar{Q}}$ implements environmental decoherence within the quantum channel, or perturbation of the channel-encoded information in the notation of \cite{knill:97}. As $H_{Q\bar{Q}} \rightarrow 0$, $Q$ and $\bar{Q}$ become decoupled.  In this case, however, the action of $H_{\bar{Q}}$ on $\mathscr{B}$ can transfer no information about $Q$.  Hence $A$ can obtain no information about $Q$ by observing qubits of $\mathscr{B}$ that are not within $q_A$.  $A$ cannot, therefore, determine by observation of either $q_A$ or any other qubits on $\mathscr{B}$ that dim($Q) \neq 0$. If we regard $q_A$ and $q_B$ as components of $A$'s and $B$'s sectors of $\mathscr{B}$ as discussed above, then in the limit as $|q_Aq_B \rangle$ approaches a pure state, i.e. decoherence approaches zero, we have $Q \rightarrow 0$; hence we also state the result as: $A$ cannot determine by observation that the quantum channel is anything other than the pure state $|q_Aq_B \rangle$.
\end{proof}

\begin{corollary} \label{qecc-cor}
The codespace dimension of a perfect QECC is operationally indistinguishable from the code dimension.
\end{corollary}

\begin{proof}
In fact, any LOCC protocol induces a QECC \cite[\S4]{fgm:24a}. More specifically, any LOCC protocol implemented jointly by Alice and Bob necessitates that $E$ implements a QECC.  The codespace dimension of this QECC is the dimension of the component $Q$ of $E$ that implements the quantum channel from Alice to Bob; the code dimension is the number of boundary qubits --- $q_A$ and $q_B$ in the above notation --- that Alice and Bob can, respectively, directly manipulate. Thus, the corollary follows directly from Theorem \ref{main-th}.
\end{proof}

In order to make the connection with ER=EPR clear, we commence with the following steps. Let us assume that Alice and Bob each implement a 3d spatial QRF, which we will denote $X_A$ and $X_B$ respectively, that allows them each to construct a representation of 3d spatial relations between identified objects from the data encoded on their accessible sectors of $\mathscr{B}$. We assume that these spatial QRFs are freely and independently chosen, and make no assumptions about the geometries they impose on Alice's and Bob's respective sectors of $\mathscr{B}$.  In particular, we make no assumption of any spatial QRF, geometry, or coordinate system shared by Alice and Bob. Hence, while they can communicate about their observed spatial relations via their classical channel, nothing requires that they can interpret each others' spatial coordinates or make any inference about their relative spatial locations.  Let us assume that they each measure a spatial location, independently and using their own spatial QRFs, at which they access their respective manipulable qubits (i.e. $q_A$ and $q_B$) of their shared quantum communication channel.  We denote these locations $x_A$ and $x_B$ respectively. The following corollary specifies the connections with EPR and ER
for each case.

\begin{corollary} \label{er-epr-cor}
In any LOCC protocol in which all systems are finite, and in which the boundary $\mathscr{B}$ between the communicating agents $A$ and $B$ and their joint environment $E$ is a holographic screen, a quantum channel implementing a shared, monogamously-entangled pair of qubits (``EPR'') is operationally indistinguishable from a topological identification of the locally-measured locations $x_A$ and $x_B$ of the qubits accessed by $A$ and $B$ respectively (``ER'').
\end{corollary}


As Corollary \ref{er-epr-cor} makes no assumption of an embedding geometry, but rather only assumes quantum and classical channels between Alice and Bob (i.e. LOCC), it substantially generalizes both the original \cite{maldecena:13} and subsequent geometric formulations of ER = EPR.

\section{Discussion}
\noindent
We derived ER=EPR as operational theorem, without assuming an overall embedding geometry --- this distinguishes our approach from the original proposal by Maldecena and Susskind.  All we require is a classical communication channel, which is a much weaker assumption. Furthermore, while it is a standard result that ER bridges are non-traversable --- see e.g. the discussion in \cite{susskind:16} --- our construction makes this particularly clear.  Suppose Alice prepares her qubit(s) $q_A$ by employing her local $z$-axis QRF, or some particular local $z$-axis QRF for each qubit in $q_A$ if multiple qubits and $z$-axis QRFs are available.  If $|q_Aq_B \rangle$ is a pure state, Alice's preparation of $q_A$ fixes $q_B$. However, Bob measures $q_B$ using his own local $z$-axis QRF(s), which bears some arbitrary relationship to Alice's.  The most Alice can do is to send Bob a classical description of her local QRFs, which requires use of a classical, i.e. causal, channel.   Alice and Bob cannot, therefore, employ their shared quantum channel for superluminal (non-causal) communication.  

Alice cannot, moreover, ``jump into'' her end of the quantum channel and ``meet with'' Bob, regardless of what Bob does.  Alice is a physical system, and so can be considered a collection of qubits in some joint state.  ``Jumping into'' the channel requires that dim(Alice) $<$ dim($q_A$).  In this case, Alice herself is the QRF in which $q_A$ is (partially) prepared. As Alice cannot determine her own complete state by observation \cite{fgl:24}, she cannot, even in principle, send Bob a classical message that would allow him to reconstruct her state from a measurement of $q_B$. No firewall is required for this negative outcome, and ``exotic'' modifications of the channel do not effect it. Hence again, our result generalizes those obtained from geometric formulations of ER = EPR. Furthermore, at least for the EPR problem, our operational approach employing QRFs, appears to be fully compatible with how Niels Bohr once envisaged the potential importance of introducing their mechanism 
\cite{dickson:04}.




\begin{thebibliography}{999}

\bibitem{maldecena:13}
Maldacena, J. and Susskind, L. Cool horizons for entangled black holes.  Fortschr. Physik 61 (2013), 781-811.

\bibitem{susskind:16}
Susskind, L. Copenhagen vs Everett, teleportation, and ER=EPR.  Fortschr. Physik 64 (2016), 551-564.

\bibitem{Almheiri:2012rt}
A.~Almheiri, D.~Marolf, J.~Polchinski and J.~Sully,
Black Holes: Complementarity or Firewalls?,
JHEP \textbf{02} (2013), 062.

\bibitem{VanRaamsdonk:2010pw}
M.~Van Raamsdonk,
Building up spacetime with quantum entanglement,
Gen. Rel. Grav. \textbf{42} (2010), 2323-2329.

\bibitem{chitambar2014everything}
Chitambar E. Leung, D., Man{\v{c}}inska, L., Ozols. M. and Winter, A. Everything you always wanted to know about LOCC (but were afraid to ask).  Commun. Math. Phys. 328 (2014), 303-326.

\bibitem{fgm:22a}
Fields, C., Glazebrook, J. F. and Marcian\`{o}, A. Sequential measurements, topological quantum field theories, and topological quantum neural networks. Fortschr. Physik 70 (2022), 202200104.

\bibitem{bartlett:07}
Bartlett, S. D., Rudolph, T. and Spekkens, R. W. Reference frames, super-selection rules, and quantum information, Rev. Mod. Phys. 79 (2007) 555-609.

\bibitem{atiyah:88}
Atiyah, M. F. Topological quantum field theory, Publ. Math. IHÉS 68 (1988) 175-186.

\bibitem{fgm:24a}
Fields, C., Glazebrook, J. F. and Marcian\`{o}, A. Communication protocols and QECCs from the perspective of TQFT, Part I: Constructing LOCC protocols and QECCs from TQFTs. Fortschr. Physik 72 (2024),  202400049.

\bibitem{fg:21}
Fields, C. and Glazebrook, J. F. Information flow in context-dependent hierarchical Bayesian inference. J. Expt. Theor. Artif. Intell. {34} (2022), 111--142.

\bibitem{fg:23}
Fields, C. and Glazebrook, J. F. Separability, contextuality, and the quantum Frame Problem.  Int. J. Theor. Phys. 62 (2023), 159.

\bibitem{tipler:14}
Tipler, F. Quantum nonlocality does not exist.  Proc. Natl. Acad. Sci. USA 111 (2014) 11281-11286.

\bibitem{zurek:09}
Zurek, W. H. Quantum Darwinism, Nature Phys. {5} (2009), 181--188.

\bibitem{tsirelson:80}
Tsirelson (Cirel'son), B.S. Quantum generalizations of Bell’s inequality. Lett. Math. Phys. 4 (1980), 93-100.

\bibitem{dai:20}
Dai, D.-C., Minic, D., Stojkovic, D. and Fu, C. Testing the ER=EPR conjecture.  Phys. Rev. D 102 (2020), 066004.

\bibitem{zanardi:00}
Zanardi, P.  Virtual quantum subsystems, Phys. Rev. Lett. 87 (2001), 077901.

\bibitem{zanardi:04}
Zanardi, P., Lidar, D. A. and Lloyd, S. Quantum tensor product structures are observable-induced, Phys. Rev. Lett. 92 (2004), 060402.


\bibitem{addazi:21}
Addazi, A., Chen, P., Fabrocini, F., Fields, C., Greco, E. Lutti, M., Marcian\`{o}, A. and Pasechnik, R.  Generalized holographic principle, gauge invariance and the emergence of gravity \`{a} la Wilczek.  {Front. Astron. Space Sci.} 8 (2021), 563450.

\bibitem{fgm:21}
Fields, C., Glazebrook, J. F. and Marcian\`{o}, A. Reference frame induced symmetry breaking on holographic screens.  {Symmetry} {13} (2021), 408.


\bibitem{fgm:22b}
Fields, C., Glazebrook, J. F. and Marcian\`{o}, A.  The physical meaning of the Holographic Principle. Quanta 11 (2022), 72-96.

\bibitem{knill:97}
Knill, E. and Laflamme, R. Theory of quantum error-correcting codes, Phys. Rev. A
55 (1997), 900-911.

\bibitem{fgl:24}
Fields, C., Glazebrook, J. F., Levin, M.  Principled limitations on self-representation for generic physical systems.  Entropy 26 (2024) 194.

\bibitem{dickson:04}
Dickson, M. Quantum reference frames in the context of EPR. {\em Philosophy of Science} {71} (2004), 655--668.





\end{thebibliography}
\end{document}